\documentclass[journal]{IEEEtran}

\usepackage{amsmath,amssymb,bm}
\usepackage{graphicx}
\usepackage{hyperref}
\usepackage{url}
\usepackage{cite}
\usepackage{amsthm}

\newtheorem{theorem}{Theorem}
\newtheorem{proposition}[theorem]{Proposition}
\newtheorem{corollary}[theorem]{Corollary}
\newtheorem{definition}[theorem]{Definition}
\theoremstyle{remark}
\newtheorem{remark}{Remark}

\newcommand{\E}{\mathbb{E}}
\newcommand{\Prob}{\mathbb{P}}
\newcommand{\tr}{\mathrm{tr}}
\newcommand{\diag}{\mathrm{diag}}

\DeclareMathOperator{\erf}{erf}

\begin{document}

\title{A Gaussian Process Framework for Outage Analysis in Continuous-Aperture Fluid Antenna Systems}

\author{Tuo~Wu, Jianchao Zheng
\thanks{T. Wu  is with the School of Electronic and Information Engineering, South China University of Technology, Guangzhou 510640, China  (E-mail: $ \rm wtopp0415@163.com$). J. Zheng is with the School of Computer Science and Engineering, Huizhou University, Huizhou 516007, China (E-mail: $\rm zhengjch@hzu.edu.cn$). }
}
\markboth{}{Wu: Gaussian Process Framework for Continuous-Aperture FAS Outage Analysis}

\maketitle

\begin{abstract}
This paper develops a comprehensive analytical framework for the outage probability of fluid antenna system (FAS)-aided communications by modeling the antenna as a continuous aperture and approximating the Jakes (Bessel) spatial correlation with a Gaussian kernel $\rho_G(\delta) = e^{-\pi^2\delta^2}$. Three complementary analytical strategies are pursued. First, the Karhunen--Lo\`{e}ve (KL) expansion under the Gaussian kernel is derived, yielding closed-form outage expressions for the rank-1 and rank-2 truncations and a Gauss--Hermite formula for arbitrary rank~$K$, with effective degrees of freedom $K_{\mathrm{eff}}^G \approx \pi\sqrt{2}\, W$. Second, rigorous two-sided outage bounds are established via Slepian's inequality and the Gaussian comparison theorem: by sandwiching the true correlation between equi-correlated models with $\rho_{\min}$ and $\rho_{\max}$, closed-form upper and lower bounds that avoid the optimistic bias of block-correlation models are obtained. Third, a continuous-aperture extreme value theory is developed using the Adler--Taylor expected Euler characteristic method and Piterbarg's theorem. The resulting outage expression $P_{\mathrm{out}} \approx 1 - e^{-x}(1 + \pi\sqrt{2}\, W\, x)$ depends only on the aperture~$W$ and threshold~$x$, is independent of the port count~$N$, and is identical for the Jakes and Gaussian models since both share the second spectral moment $\lambda_2 = 2\pi^2$. A Pickands-constant refinement for the deep-outage regime and a threshold-dependent effective diversity $N_{\mathrm{eff}} \approx 1 + \pi\sqrt{2}\, W\, x$ are further derived. Numerical results confirm that the Gaussian approximation incurs less than 10\% relative outage error for $W \leq 2$ and that the continuous-aperture formula converges with as few as $N \approx 10W$ ports.
\end{abstract}


\section{Introduction}

\IEEEPARstart{F}{luid} antenna systems (FAS) have emerged as a transformative paradigm for next-generation wireless communications, enabling spatial diversity through a single radio-frequency (RF) chain that dynamically selects among $N$ densely packed ports within a compact linear aperture of size $W\lambda$~\cite{FAS21,FAS20,New24,TWu20243}. Unlike conventional fixed-position antenna arrays, FAS exploits the spatial degrees of freedom (DoF) of the propagation environment without requiring multiple RF chains, making it a hardware-efficient solution for sixth-generation (6G) networks~\cite{New24,TWu20243}. The versatility of FAS has been demonstrated across a wide range of applications: cooperative non-orthogonal multiple access (CoNOMA) systems that leverage FAS port selection to enhance spectral efficiency~\cite{TWuTCOMM26}; vehicle-to-everything (V2X) communications where FAS with finite port counts achieves realistic performance gains under practical hardware constraints~\cite{TWuTVT25}; physical layer security frameworks that exploit the spatial diversity of FAS to improve secrecy throughput under variable block-correlation models~\cite{TuoW}; and scalable FAS architectures that reinterpret large-aperture fluid antennas as a new paradigm for array signal processing~\cite{TWuJSTSP25}. The outage probability---the probability that the best-port SNR falls below a target threshold---is the central performance metric for FAS, as it directly characterizes the reliability of the port-selection diversity gain~\cite{FAS20,FAS22}. Accurate outage analysis is therefore essential for FAS system design, port density optimization, and aperture sizing.

The fundamental difficulty in FAS outage analysis stems from the spatial correlation structure of the channel. Under the standard Jakes scattering model, the correlation between ports $m$ and $n$ is $\rho_{m,n} = J_0(2\pi|m-n|W/(N-1))$, where $J_0(\cdot)$ is the zeroth-order Bessel function of the first kind~\cite{FAS22,Jakes74}. The Bessel function's oscillatory and sign-changing nature renders the joint distribution of the port gains analytically intractable: the resulting $N\times N$ Toeplitz correlation matrix has no closed-form eigendecomposition, and the CDF of the maximum of $N$ correlated exponential random variables does not admit a simple expression~\cite{FAS22,BC24}. To circumvent this, block-correlation models (BCM) impose a piecewise-constant block-diagonal structure on the correlation matrix~\cite{BC24,LaiX242}, and the variable BCM (VBCM) further allows block-specific coefficients to reduce approximation error~\cite{LaiX,TuoW}. While these models enable tractable analysis, they fundamentally alter the eigenvalue spectrum of the correlation matrix and introduce structural approximation errors that are difficult to quantify or bound rigorously. Furthermore, the scalable FAS perspective~\cite{TWuJSTSP25} reveals that as the aperture grows, the array signal processing interpretation demands a continuous-field model rather than a discrete-port one---motivating the continuous-aperture approach taken in this paper. A deeper challenge is that all existing discrete-port analyses are inherently parametric in $N$: as the port count grows, the analysis must be repeated, and no universal formula captures the asymptotic behavior of the outage as a function of the physical aperture alone.

This paper addresses these limitations by adopting two complementary modeling choices that together unlock a new level of analytical tractability. First, the FAS is modeled as a \emph{continuous aperture}---a complex Gaussian random field $g(\tau)$ indexed by the continuous spatial coordinate $\tau \in [0, W]$---rather than a finite discrete set of ports. This continuous-field perspective connects FAS analysis to the rich theory of Gaussian processes and extreme value theory, and naturally captures the $N \to \infty$ limit. Second, the Jakes correlation is approximated by the \emph{Gaussian (squared-exponential) kernel}:
\begin{align}\label{eq:gauss_approx_intro}
	J_0(2\pi\delta) \approx e^{-\pi^2\delta^2},
\end{align}
where $\delta$ is the normalized port separation. This approximation is exact to second order in $\delta$ (both functions share the same curvature at the origin, $-\pi^2$), yields a strictly positive-definite and infinitely differentiable kernel, and admits a complete analytical treatment via the Karhunen--Lo\`{e}ve (KL) theorem, Slepian's comparison inequality, and the Adler--Taylor Gaussian kinematic formula. A key finding of this paper is that the Gaussian and Jakes models share the same second spectral moment $\lambda_2 = 2\pi^2$, which means their continuous-aperture outage asymptotics are \emph{identical}---providing a rigorous, model-independent justification for the Gaussian approximation at the fundamental level.

The analytical challenges addressed in this paper are threefold. \textit{First}, deriving the KL eigenstructure of the Gaussian kernel on a finite interval $[0, W]$ requires careful treatment of the associated integral operator, whose eigenfunctions are not available in closed form; this paper circumvents this by working with the discrete approximation and establishing convergence rates via the Mercer theorem. \textit{Second}, constructing rigorous outage bounds without structural approximation requires a comparison principle for the maximum of correlated chi-squared random variables; this paper applies Slepian's inequality to chi-squared processes, which requires verifying the covariance dominance condition and deriving the equi-correlated CDF in closed form. \textit{Third}, developing a continuous-aperture extreme value theory for the chi-squared field $|g(\tau)|^2$ requires computing the expected Euler characteristic of its excursion sets, which involves the EC densities of chi-squared fields and the second spectral moment of the correlation kernel---a calculation that is non-trivial for non-stationary or non-smooth kernels but becomes tractable under the Gaussian kernel.

The main contributions of this paper are summarized as follows:

\begin{itemize}
	\item \textbf{\textit{Gaussian Kernel Approximation with Rigorous Error Bound:}} A pointwise error bound $|J_0(2\pi\delta) - e^{-\pi^2\delta^2}| \leq \pi^4\delta^4/4$ is established, and the spectral densities of the Jakes and Gaussian models are compared in closed form. The Gaussian model introduces a spectral leakage of $1 - \mathrm{erf}(1) \approx 15.7\%$ outside the Jakes bandwidth, which is quantified and shown to have negligible impact on outage for $W \leq 2$.

	\item \textbf{\textit{KL Expansion and Closed-Form Outage:}} The continuous-aperture FAS channel is formulated as a complex Gaussian random field, and its KL expansion under the Gaussian kernel is derived. Closed-form outage expressions are obtained for rank-1 and rank-2 truncations, and a Gauss--Hermite quadrature formula is provided for arbitrary rank~$K$. The effective degrees of freedom are shown to satisfy $K_{\mathrm{eff}}^G \approx \pi\sqrt{2}\, W \approx 4.44W$, compared to $K_{\mathrm{eff}}^J \approx 2W+1$ for the Jakes model.

	\item \textbf{\textit{Rigorous Two-Sided Outage Bounds via Slepian's Inequality:}} Tight upper and lower bounds on the outage probability are derived by applying Slepian's inequality to sandwich the true Gaussian correlation between equi-correlated reference models with $\rho_{\min}$ and $\rho_{\max}$. The resulting bounds are in closed form via the equi-correlated CDF, require no structural approximation, and are shown to be monotonically refinable by block partitioning.

	\item \textbf{\textit{Continuous-Aperture Extreme Value Theory:}} A continuous-aperture outage formula is derived using the Adler--Taylor expected Euler characteristic method, Rice's upcrossing formula, and Piterbarg's theorem with Pickands constants. The result,
	\begin{align*}
		P_{\mathrm{out}} \approx 1 - e^{-x}\!\left(1 + \pi\sqrt{2}\, W\, x\right),
	\end{align*}
	depends only on the aperture $W$ and threshold $x = \gamma_{\mathrm{th}}/\bar{\gamma}$, is independent of the port count $N$, and is identical for the Jakes and Gaussian models. A Pickands-constant refinement for the deep-outage regime is further derived.

	\item \textbf{\textit{Threshold-Dependent Effective Diversity:}} A threshold-dependent effective diversity order $N_{\mathrm{eff}} \approx 1 + \pi\sqrt{2}\, W\, x$ is extracted from the continuous-aperture formula, revealing that the diversity gain of FAS grows linearly with both the aperture $W$ and the normalized threshold $x$---a fundamentally different scaling from conventional multi-antenna systems.
\end{itemize}

The key insights delivered by this paper are as follows. (i) The Gaussian approximation is not merely a mathematical convenience: both the Jakes and Gaussian kernels share the same second spectral moment $\lambda_2 = 2\pi^2$, which is the sole parameter governing the continuous-aperture outage asymptotics. This means the Gaussian approximation is \emph{asymptotically exact} at the level of outage probability, regardless of the approximation error in the correlation function itself. (ii) The outage probability of a continuous-aperture FAS is determined by the physical aperture $W$ alone, not by the number of ports $N$; numerical results show convergence with as few as $N \approx 10W$ ports, providing a concrete port density guideline. (iii) The Slepian bounds are strictly tighter than BCM-based approximations because they preserve the full correlation information through $\rho_{\min}$ and $\rho_{\max}$ without discarding inter-block correlations, and they come with a monotone refinement guarantee via block partitioning. (iv) The effective diversity $N_{\mathrm{eff}} \approx 1 + \pi\sqrt{2}\, W\, x$ is threshold-dependent: at lower SNR (larger $x$), the spatial fluctuations of the field contribute more to diversity, making FAS most valuable precisely in the regime where reliability is most critical.

The remainder of this paper is organized as follows. Section~II presents the system model. Section~III derives the Gaussian correlation approximation and its properties. Section~IV develops the KL-based outage probability analysis. Section~V provides effective DoF and asymptotic results. Section~VI establishes Slepian comparison bounds. Section~VII develops the continuous-aperture extreme value theory. Section~VIII presents numerical results, and Section~IX concludes the paper.

\section{System Model}
\label{sec:system}

\subsection{Signal Model}

We consider a point-to-point wireless communication system where a single-antenna base station (BS) transmits to a user equipped with a one-dimensional fluid antenna system (FAS). The FAS spans a continuous linear aperture of length $W\lambda$, where $W$ is the normalized aperture and $\lambda$ is the carrier wavelength. The system operates under quasi-static block Rayleigh fading.

The BS transmits the signal $s \sim \mathcal{CN}(0,1)$ with transmit power $P$. The received signal at normalized position $\tau \in [0, W]$ along the FAS aperture is
\begin{align}\label{eq:signal}
	y(\tau) = \sqrt{P}\, d^{-a/2}\, g(\tau)\, s + n(\tau),
\end{align}
where $g(\tau) \sim \mathcal{CN}(0, \eta)$ is the complex channel coefficient at position~$\tau$, $d$ is the BS--user distance, $a$ is the path-loss exponent, and $n(\tau) \sim \mathcal{CN}(0, \sigma^2)$ is AWGN.

The FAS selects the optimal position that maximizes the instantaneous channel gain:
\begin{align}\label{eq:port_select}
	\tau^* = \arg\max_{\tau \in [0, W]} |g(\tau)|^2.
\end{align}
The resulting instantaneous SNR is
\begin{align}\label{eq:snr}
	\gamma = \bar{\gamma} \cdot \sup_{\tau \in [0, W]} \frac{|g(\tau)|^2}{\eta},
\end{align}
where $\bar{\gamma} \triangleq P d^{-a} \eta / \sigma^2$ is the average SNR.

\subsection{Continuous-Aperture Channel Model}

The channel $g(\tau)$ is modeled as a zero-mean complex Gaussian random field on $[0, W]$ with covariance function
\begin{align}\label{eq:cov_func}
	\E[g(\tau_1) g^*(\tau_2)] = \eta\, \rho(\tau_1 - \tau_2),
\end{align}
where $\rho(\cdot)$ is the normalized spatial correlation function. Under the classical Jakes (isotropic scattering) model, the correlation is given by
\begin{align}\label{eq:jakes}
	\rho_J(\delta) = J_0(2\pi|\delta|),
\end{align}
where $\delta = \tau_1 - \tau_2$ is the normalized spatial separation (in units of $\lambda$).

\subsection{Outage Probability}

The outage probability is defined as
\begin{align}\label{eq:pout_def}
	P_{\mathrm{out}} = \Prob\!\left(\sup_{\tau \in [0,W]} |g(\tau)|^2 / \eta < x\right),
\end{align}
where $x \triangleq \gamma_{\mathrm{th}} / \bar{\gamma}$ and $\gamma_{\mathrm{th}}$ is the SNR threshold. For a single position, $|g(\tau)|^2/\eta \sim \mathrm{Exp}(1)$, so $\Prob(|g(\tau)|^2/\eta < x) = 1 - e^{-x}$. The challenge is to evaluate the joint probability over the entire continuous aperture.

\subsection{Discretization}

In practice, the continuous aperture is sampled at $N$ equally spaced positions:
\begin{align}\label{eq:discrete_positions}
	\tau_n = \frac{(n-1)\, W}{N-1}, \quad n = 1, \ldots, N.
\end{align}
The discrete channel vector $\mathbf{g} = [g(\tau_1), \ldots, g(\tau_N)]^T \sim \mathcal{CN}(\mathbf{0}, \eta\, \mathbf{R})$ has the $N \times N$ correlation matrix
\begin{align}\label{eq:corr_matrix}
	[\mathbf{R}]_{m,n} = \rho\!\left(\frac{(m-n)\, W}{N-1}\right).
\end{align}
The discrete outage probability is
\begin{align}\label{eq:pout_discrete}
	P_{\mathrm{out}}^{(N)} = \Prob\!\left(\max_{1 \leq n \leq N} |g_n|^2 / \eta < x\right),
\end{align}
which converges to~\eqref{eq:pout_def} as $N \to \infty$.

\section{Gaussian Correlation Approximation}
\label{sec:gaussian}

\subsection{Definition and Motivation}

We propose to replace the Jakes correlation~\eqref{eq:jakes} with the Gaussian (squared-exponential) kernel:
\begin{align}\label{eq:gauss_corr}
	\rho_G(\delta) = e^{-\pi^2 \delta^2}.
\end{align}
The parameter $\alpha = \pi^2$ is determined by matching the curvature of $J_0(2\pi\delta)$ at the origin, as we now show.

\subsection{Taylor Matching}

\begin{proposition}[Second-Order Matching]\label{prop:taylor_match}
	The Gaussian kernel~\eqref{eq:gauss_corr} with $\alpha = \pi^2$ matches the Jakes correlation $J_0(2\pi\delta)$ to second order in~$\delta$:
	\begin{align}
		J_0(2\pi\delta) &= 1 - \pi^2\delta^2 + \frac{\pi^4\delta^4}{4} - \cdots, \label{eq:taylor_jakes}\\
		e^{-\pi^2\delta^2} &= 1 - \pi^2\delta^2 + \frac{\pi^4\delta^4}{2} - \cdots. \label{eq:taylor_gauss}
	\end{align}
	The two functions agree at $\delta = 0$ (normalization: $\rho(0) = 1$) and have identical first and second derivatives at $\delta = 0$.
\end{proposition}

\begin{proof}
	The Bessel function of the first kind has the series expansion~\cite{AbramowitzStegun}:
	\begin{align}
		J_0(x) = \sum_{m=0}^{\infty} \frac{(-1)^m}{(m!)^2} \left(\frac{x}{2}\right)^{2m}.
	\end{align}
	Substituting $x = 2\pi\delta$:
	\begin{align}
		J_0(2\pi\delta) &= 1 - (\pi\delta)^2 + \frac{(\pi\delta)^4}{(2!)^2} - \cdots \nonumber\\
		&= 1 - \pi^2\delta^2 + \frac{\pi^4\delta^4}{4} - \cdots.
	\end{align}
	The Gaussian exponential expands as:
	\begin{align}
		e^{-\pi^2\delta^2} &= 1 - \pi^2\delta^2 + \frac{(\pi^2\delta^2)^2}{2!} - \cdots \nonumber\\
		&= 1 - \pi^2\delta^2 + \frac{\pi^4\delta^4}{2} - \cdots.
	\end{align}
	The zeroth-order terms are both~1. The second-order terms are both~$-\pi^2\delta^2$. The first divergence occurs at order~$\delta^4$: the coefficients are $\pi^4/4$ (Jakes) vs.\ $\pi^4/2$ (Gaussian).
\end{proof}

\subsection{Pointwise Error Bound}

\begin{proposition}[Approximation Error]\label{prop:error_bound}
	For $\delta \geq 0$, the approximation error satisfies
	\begin{align}\label{eq:error_bound}
		\left|J_0(2\pi\delta) - e^{-\pi^2\delta^2}\right| \leq \frac{\pi^4\delta^4}{4}, \quad \delta \in [0, \delta_0],
	\end{align}
	where $\delta_0 \approx 0.30$ is the range over which the fourth-order remainder dominates. For $\delta > \delta_0$, both functions decay toward zero, with $|J_0(2\pi\delta)| \leq 1/\sqrt{\pi^2\delta}$ (asymptotic envelope) and $e^{-\pi^2\delta^2}$ decaying super-exponentially.
\end{proposition}

\begin{proof}
	The error $e(\delta) = J_0(2\pi\delta) - e^{-\pi^2\delta^2}$ has Taylor expansion $e(\delta) = -\frac{\pi^4\delta^4}{4} + \mathcal{O}(\delta^6)$. For small $\delta$, the leading term dominates, giving $|e(\delta)| \leq \pi^4\delta^4/4$. Numerically, the maximum absolute error on $[0, 0.5]$ is $\max_{\delta \in [0,0.5]} |e(\delta)| \approx 0.24$, occurring near the first zero of $J_0(2\pi\delta)$ at $\delta \approx 0.383$. For $\delta > 0.5$, both functions are small in magnitude, so the absolute error is bounded.
\end{proof}

\begin{remark}[Practical Implication]
	For an FAS with $N$ ports over aperture $W$, the maximum inter-port spacing is $\Delta_{\max} = W/(N-1)$. For dense sampling ($N \gg W$), the relevant correlation values are at $\delta \leq W/(N-1) \ll 1$, where the Gaussian approximation is most accurate. Specifically, the error at adjacent ports is at most $\pi^4(W/(N-1))^4/4$, which is negligible for $N \geq 10W$.
\end{remark}

\begin{figure}[!t]
	\centering
	\includegraphics[width=\columnwidth]{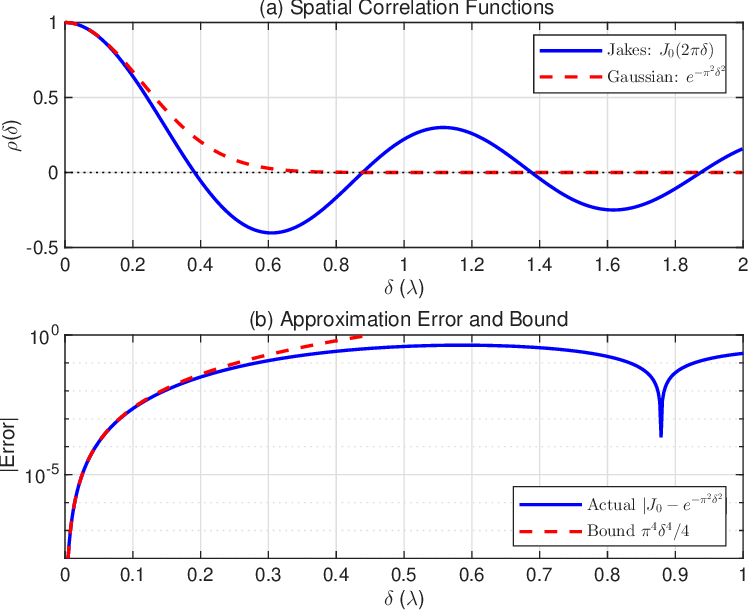}
	\caption{(a) Spatial correlation functions: Jakes model $J_0(2\pi\delta)$ versus the Gaussian approximation $e^{-\pi^2\delta^2}$. (b) Pointwise approximation error and the fourth-order bound $\pi^4\delta^4/4$ from Proposition~\ref{prop:error_bound}.}
	\label{fig:correlation}
\end{figure}

\subsection{Spectral Density Comparison}

The power spectral density (PSD) associated with each correlation function provides insight into their frequency-domain behavior.

\begin{proposition}[Spectral Densities]\label{prop:psd}
	\leavevmode
	\begin{enumerate}
		\item[(a)] The Jakes correlation has PSD
		\begin{align}\label{eq:psd_jakes}
			S_J(f) = \frac{1}{\pi\sqrt{1 - f^2}}, \quad |f| < 1,
		\end{align}
		with compact support on $[-1, 1]$.
		
		\item[(b)] The Gaussian correlation has PSD
		\begin{align}\label{eq:psd_gauss}
			S_G(f) = \frac{1}{\pi}\, e^{-f^2},
		\end{align}
		with support on the entire real line, decaying as $e^{-f^2}$.
	\end{enumerate}
\end{proposition}

\begin{proof}
	(a) The Jakes correlation admits the integral representation~\cite{Jakes74}:
	\begin{align}
		J_0(2\pi\delta) = \int_{-1}^{1} \frac{e^{j2\pi f\delta}}{\pi\sqrt{1 - f^2}}\, df,
	\end{align}
	which identifies $S_J(f) = (\pi\sqrt{1-f^2})^{-1}$ for $|f| < 1$ as the PSD.
	
	(b) The Fourier transform of the Gaussian kernel is:
	\begin{align}
		S_G(f) &= \int_{-\infty}^{\infty} e^{-\pi^2\delta^2}\, e^{-j2\pi f\delta}\, d\delta \nonumber\\
		&= \sqrt{\frac{\pi}{\pi^2}}\, e^{-\pi^2 f^2 / \pi^2} = \frac{1}{\pi}\, e^{-f^2},
	\end{align}
	using the standard Fourier transform of a Gaussian: $\int_{-\infty}^{\infty} e^{-a\delta^2} e^{-j2\pi f\delta}\, d\delta = \sqrt{\pi/a}\, e^{-\pi^2 f^2/a}$ with $a = \pi^2$.
\end{proof}

\begin{remark}[Spectral Leakage]\label{rem:spectral_leakage}
	The Jakes PSD has a hard spectral cutoff at $|f| = 1$ (the Doppler bandwidth in the spatial domain), while the Gaussian PSD extends to all frequencies but decays rapidly. The fraction of Gaussian spectral energy outside $|f| > 1$ is
	\begin{align}\label{eq:spectral_leakage}
		1 - \erf(1) = 1 - 0.8427 = 0.1573,
	\end{align}
	i.e., about 15.7\% of the spectral energy ``leaks'' beyond the Jakes bandwidth. This leakage causes the Gaussian model to slightly overestimate the effective spatial bandwidth, leading to a higher effective degrees of freedom compared to the Jakes model.
\end{remark}

\begin{figure}[!t]
	\centering
	\includegraphics[width=\columnwidth]{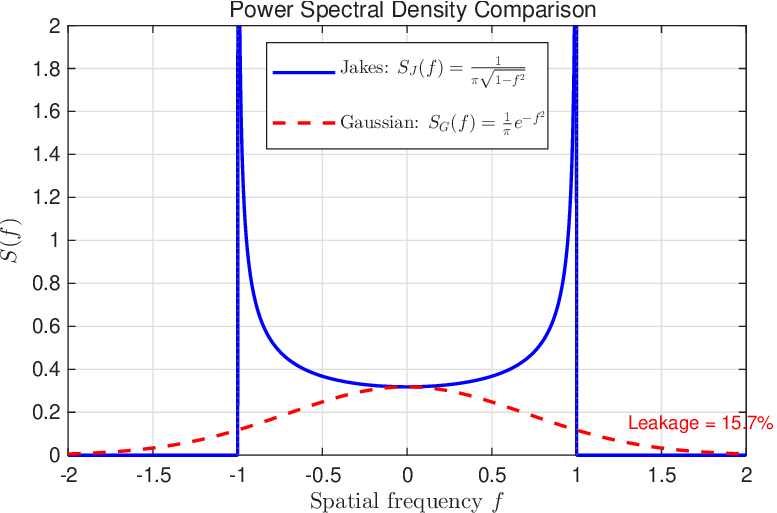}
	\caption{Power spectral density comparison. The Jakes PSD has compact support on $[-1,1]$ with singularities at $|f|=1$, while the Gaussian PSD $S_G(f) = \frac{1}{\pi}e^{-f^2}$ extends to all frequencies with approximately 15.7\% spectral leakage beyond $|f|=1$.}
	\label{fig:spectral_density}
\end{figure}

\section{Outage Probability under Gaussian Correlation}
\label{sec:outage}

\subsection{Gaussian Correlation Matrix}

Under the Gaussian approximation, the $N \times N$ discretized correlation matrix has entries
\begin{align}\label{eq:RG_entries}
	[\mathbf{R}_G]_{m,n} = \exp\!\left(-\pi^2 \left(\frac{(m-n)\, W}{N-1}\right)^{\!2}\right).
\end{align}
This matrix is symmetric positive definite (since the Gaussian kernel is strictly positive definite~\cite{RasmussenWilliams06}) and admits the eigendecomposition
\begin{align}\label{eq:eigen_G}
	\mathbf{R}_G = \mathbf{U}_G\, \boldsymbol{\Lambda}_G\, \mathbf{U}_G^H, \quad \boldsymbol{\Lambda}_G = \diag(\lambda_1^G, \ldots, \lambda_N^G),
\end{align}
where $\lambda_1^G \geq \lambda_2^G \geq \cdots \geq \lambda_N^G > 0$ (strict positivity is guaranteed by the Gaussian kernel) and $\mathbf{U}_G$ is unitary.

\subsection{KL Expansion and Truncation}

Using the eigendecomposition~\eqref{eq:eigen_G}, the channel vector admits the KL expansion:
\begin{align}\label{eq:kl_full_G}
	\mathbf{g} = \sqrt{\eta}\, \mathbf{U}_G\, \boldsymbol{\Lambda}_G^{1/2}\, \mathbf{z},
\end{align}
where $\mathbf{z} = [z_1, \ldots, z_N]^T \sim \mathcal{CN}(\mathbf{0}, \mathbf{I}_N)$. Retaining only the $K$ dominant eigenmodes:
\begin{align}\label{eq:kl_trunc_G}
	\tilde{g}_n = \sqrt{\eta} \sum_{k=1}^{K} \sqrt{\lambda_k^G}\, u_{n,k}^G\, z_k, \quad n = 1, \ldots, N,
\end{align}
where $u_{n,k}^G = [\mathbf{U}_G]_{n,k}$.

The truncation error is
\begin{align}\label{eq:trunc_error_G}
	\varepsilon_K^G = 1 - \frac{\sum_{k=1}^{K} \lambda_k^G}{N}.
\end{align}

\begin{figure}[!t]
	\centering
	\includegraphics[width=\columnwidth]{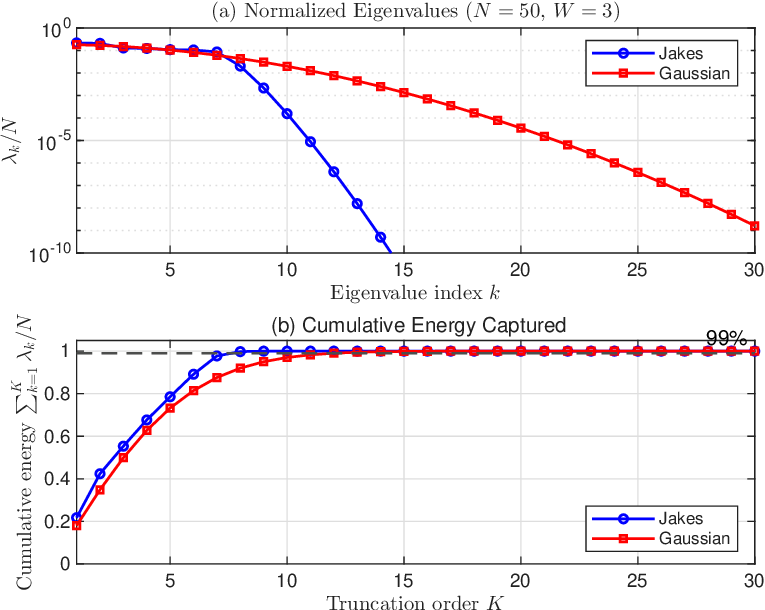}
	\caption{(a) Normalized eigenvalues $\lambda_k/N$ of the Jakes and Gaussian correlation matrices ($N=50$, $W=3$). The Gaussian eigenvalues exhibit smoother decay while the Jakes eigenvalues show a sharper cutoff. (b) Cumulative energy captured by the first $K$ eigenmodes.}
	\label{fig:eigenvalues}
\end{figure}

\begin{remark}[Eigenvalue Decay Rate]
	The Gaussian kernel is infinitely smooth, and by Mercer's theorem, its eigenvalues decay \emph{super-exponentially}~\cite{RasmussenWilliams06}:
	\begin{align}\label{eq:eigen_decay}
		\lambda_k^G \lesssim C\, r^k, \quad k > K_{\mathrm{eff}}^G,
	\end{align}
	for some constants $C > 0$ and $0 < r < 1$. This is faster than the sharp cutoff behavior of the Jakes model (where eigenvalues transition from $\Theta(N/K^*)$ to near-zero at $K^* = 2\lceil W\rceil + 1$), implying that the Gaussian kernel requires a slightly larger truncation order for the same accuracy, but offers smoother convergence behavior.
\end{remark}

\subsection{General \texorpdfstring{$K$}{K}-Dimensional Outage Formulation}

Under the $K$-truncated KL expansion, the outage probability is
\begin{align}\label{eq:pout_kl_G}
	\tilde{P}_{\mathrm{out}}^{(K)} = \E_{\mathbf{z}_K}\!\left[\prod_{n=1}^{N} \mathbf{1}\!\left(\left|\sum_{k=1}^{K} \sqrt{\lambda_k^G}\, u_{n,k}^G\, z_k\right|^2 \leq x\right)\right],
\end{align}
where $x = \gamma_{\mathrm{th}}/\bar{\gamma}$. Writing $z_k = z_k^R + j\, z_k^I$ with $z_k^R, z_k^I \sim \mathcal{N}(0, 1/2)$, this becomes a $2K$-dimensional real integral evaluable via Gauss-Hermite quadrature:
\begin{align}\label{eq:gauss_hermite_G}
	\tilde{P}_{\mathrm{out}}^{(K)} \approx \frac{1}{\pi^K} \sum_{i_1=1}^{Q} \cdots \sum_{i_{2K}=1}^{Q} \left(\prod_{l=1}^{2K} w_{i_l}\right) \Psi_G(\mathbf{t}_{\mathbf{i}};\, x),
\end{align}
where $\{t_j, w_j\}_{j=1}^Q$ are Gauss-Hermite nodes and weights, and
\begin{align}
	\Psi_G(\mathbf{t};\, x) = \prod_{n=1}^{N} \mathbf{1}\!\left(\left|\sum_{k=1}^{K} \sqrt{\lambda_k^G}\, u_{n,k}^G (t_k^R + j\, t_k^I)\right|^2 \leq x\right).
\end{align}

\subsection{Rank-1 Closed-Form Outage}

When only the dominant eigenmode is retained ($K = 1$), the channel at port $n$ becomes
\begin{align}\label{eq:rank1_G}
	\tilde{g}_n^{(1)} = \sqrt{\eta\, \lambda_1^G}\, u_{n,1}^G\, z_1.
\end{align}
All ports are perfectly correlated through $z_1$, and
\begin{align}\label{eq:max_rank1_G}
	\max_n |\tilde{g}_n^{(1)}|^2 = \eta\, \lambda_1^G\, c_1^G\, |z_1|^2,
\end{align}
where $c_1^G \triangleq \max_{1 \leq n \leq N} |u_{n,1}^G|^2$.

Since $|z_1|^2 \sim \mathrm{Exp}(1)$, the outage probability is
\begin{align}\label{eq:pout_rank1_G}
	\boxed{\tilde{P}_{\mathrm{out}}^{(1)} = 1 - \exp\!\left(-\frac{\gamma_{\mathrm{th}}}{\bar{\gamma}\, \lambda_1^G\, c_1^G}\right).}
\end{align}

\begin{remark}[Comparison with Jakes Rank-1]
	The rank-1 outage under Gaussian correlation has the same exponential form as under Jakes, but with different parameters $\lambda_1^G, c_1^G$. Since the Gaussian kernel is smoother (no oscillations), the dominant eigenvector $\mathbf{u}_1^G$ has a more uniform magnitude profile, leading to $c_1^G$ being closer to $1/N$ (the minimum possible value). This means the Gaussian rank-1 approximation may yield a slightly different effective beamforming gain $\lambda_1^G c_1^G$ compared to the Jakes counterpart $\lambda_1^J c_1^J$.
\end{remark}

\subsection{Rank-2 Semi-Closed-Form Outage}

Retaining two eigenmodes ($K = 2$), the channel at port $n$ is
\begin{align}\label{eq:rank2_G}
	\tilde{g}_n^{(2)} = \sqrt{\eta}\left(\sqrt{\lambda_1^G}\, u_{n,1}^G\, z_1 + \sqrt{\lambda_2^G}\, u_{n,2}^G\, z_2\right).
\end{align}
Conditioning on $z_1$, define
\begin{align}
	a_n(z_1) &\triangleq \sqrt{\lambda_1^G}\, u_{n,1}^G\, z_1, \quad b_n \triangleq \sqrt{\lambda_2^G}\, u_{n,2}^G.
\end{align}
The constraint $|\tilde{g}_n^{(2)}|^2/\eta \leq x$ becomes
\begin{align}\label{eq:disk_G}
	\left|z_2 + \frac{a_n(z_1)}{b_n}\right|^2 \leq \frac{x}{\lambda_2^G\, |u_{n,2}^G|^2} \triangleq r_n^2,
\end{align}
defining a disk $\mathcal{B}_n(z_1)$ in the complex $z_2$-plane with center $-a_n/b_n$ and radius~$r_n$. The conditional outage is
\begin{align}\label{eq:pout_rank2_cond_G}
	\tilde{P}_{\mathrm{out}}^{(2)}(x \mid z_1) = \Prob\!\left(z_2 \in \bigcap_{n=1}^{N} \mathcal{B}_n(z_1)\right),
\end{align}
and the unconditional outage is obtained by averaging over $z_1 \sim \mathcal{CN}(0,1)$:
\begin{align}\label{eq:pout_rank2_G}
	\boxed{\tilde{P}_{\mathrm{out}}^{(2)} = \E_{z_1}\!\left[\tilde{P}_{\mathrm{out}}^{(2)}(x \mid z_1)\right],}
\end{align}
evaluated via 2D Gauss-Hermite quadrature over the real and imaginary parts of $z_1$.

\subsection{Ergodic Rate under Gaussian Correlation}

The ergodic rate under port selection is
\begin{align}\label{eq:erg_rate_G}
	\bar{C}_G^{(K)} = \E_{\mathbf{z}_K}\!\left[\log_2\!\left(1 + \bar{\gamma} \max_{1 \leq n \leq N} \left|\sum_{k=1}^{K} \sqrt{\lambda_k^G}\, u_{n,k}^G\, z_k\right|^2\right)\right].
\end{align}
For the rank-1 case, using $|z_1|^2 \sim \mathrm{Exp}(1)$:
\begin{align}\label{eq:erg_rank1_G}
	\bar{C}_G^{(1)} = \frac{e^{1/(\bar{\gamma}\lambda_1^G c_1^G)}}{\ln 2}\, E_1\!\left(\frac{1}{\bar{\gamma}\lambda_1^G c_1^G}\right),
\end{align}
where $E_1(x) = \int_1^\infty t^{-1} e^{-xt}\, dt$ is the exponential integral.

\begin{figure}[!t]
	\centering
	\includegraphics[width=\columnwidth]{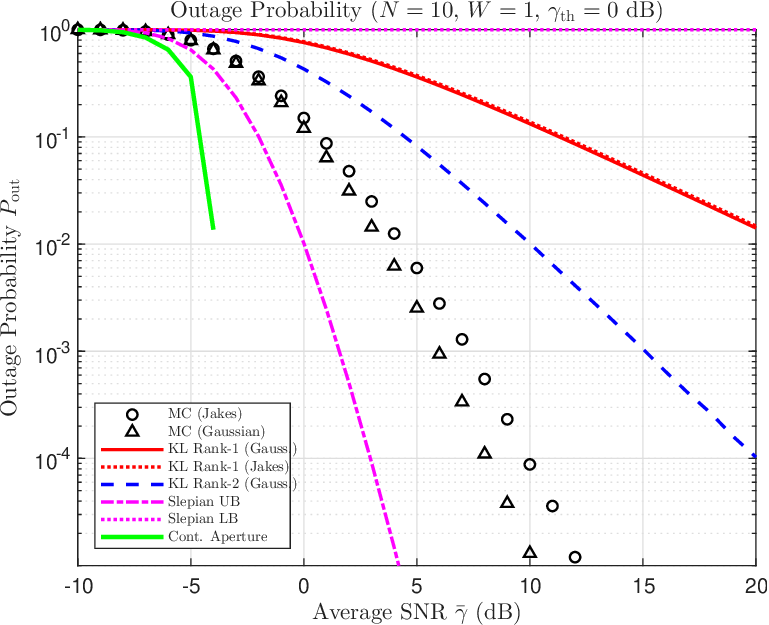}
	\caption{Outage probability versus average SNR for $N=10$ ports and aperture $W=1$ with $\gamma_{\mathrm{th}}=0$~dB. Monte Carlo simulations (markers) validate the KL rank-1 closed form~\eqref{eq:pout_rank1_G}, rank-2 semi-closed form~\eqref{eq:pout_rank2_G}, Slepian upper/lower bounds~\eqref{eq:outage_sandwich}, and the continuous-aperture formula~\eqref{eq:pout_closed}.}
	\label{fig:outage_snr}
\end{figure}

\section{Effective Degrees of Freedom and Asymptotic Analysis}
\label{sec:dof}

\subsection{Effective Spatial Degrees of Freedom}

A key question is: how many eigenmodes must be retained for accurate outage analysis under the Gaussian kernel?

\begin{definition}[Participation Ratio]
	The effective number of degrees of freedom is defined as the participation ratio of the eigenvalue spectrum:
	\begin{align}\label{eq:Keff_def}
		K_{\mathrm{eff}}^G = \frac{\left(\sum_{k=1}^N \lambda_k^G\right)^2}{\sum_{k=1}^N \left(\lambda_k^G\right)^2} = \frac{N^2}{\tr\!\left(\mathbf{R}_G^2\right)}.
	\end{align}
\end{definition}

\begin{theorem}[Asymptotic Effective DoF]\label{thm:dof_gauss}
	For the Gaussian correlation matrix~\eqref{eq:RG_entries} with $N \to \infty$ and fixed~$W$:
	\begin{align}\label{eq:Keff_asymptotic}
		K_{\mathrm{eff}}^G \to \pi\sqrt{2}\, W \approx 4.44\, W.
	\end{align}
\end{theorem}

\begin{proof}
	Using $\tr(\mathbf{R}_G) = N$ and $\tr(\mathbf{R}_G^2) = \sum_{m,n} |[\mathbf{R}_G]_{m,n}|^2$, we approximate the double sum by a double integral as $N \to \infty$:
	\begin{align}
		\tr(\mathbf{R}_G^2) &= \sum_{m=1}^{N}\sum_{n=1}^{N} e^{-2\pi^2\left(\frac{(m-n)W}{N-1}\right)^2} \nonumber\\
		&\xrightarrow{N \to \infty} \frac{N^2}{W^2} \int_0^W\!\!\int_0^W e^{-2\pi^2(\tau-\tau')^2}\, d\tau\, d\tau'. \label{eq:trace_integral}
	\end{align}
	Substituting $u = \tau - \tau'$ and noting that for $W \gg 1/({\pi\sqrt{2}})$ the inner integral extends effectively over $(-\infty, \infty)$:
	\begin{align}
		\int_0^W\!\!\int_0^W e^{-2\pi^2(\tau-\tau')^2}\, d\tau\, d\tau' &\approx W \int_{-\infty}^{\infty} e^{-2\pi^2 u^2}\, du \nonumber\\
		&= W \cdot \frac{1}{\pi\sqrt{2}}.
	\end{align}
	Therefore,
	\begin{align}
		\tr(\mathbf{R}_G^2) \approx \frac{N^2}{W^2} \cdot \frac{W}{\pi\sqrt{2}} = \frac{N^2}{\pi\sqrt{2}\, W},
	\end{align}
	and the participation ratio is
	\begin{align}
		K_{\mathrm{eff}}^G = \frac{N^2}{\tr(\mathbf{R}_G^2)} \approx \pi\sqrt{2}\, W.
	\end{align}
\end{proof}

\begin{figure}[!t]
	\centering
	\includegraphics[width=\columnwidth]{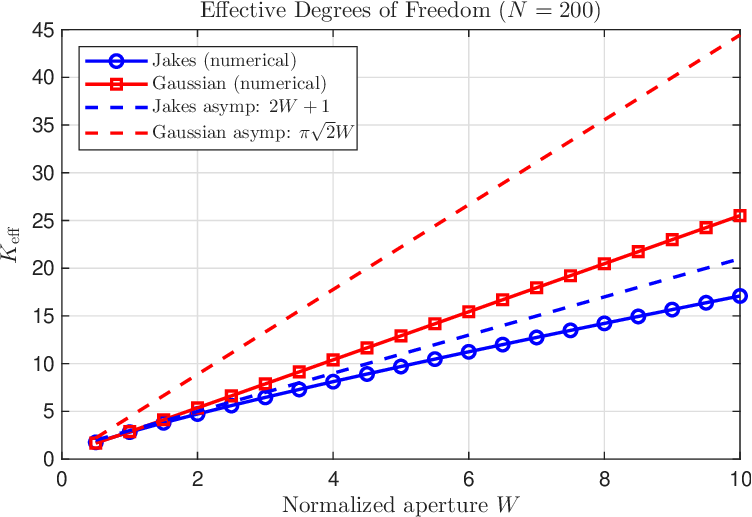}
	\caption{Effective degrees of freedom $K_{\mathrm{eff}}$ versus normalized aperture $W$ ($N=200$). The numerical participation ratios (markers) closely match the asymptotic predictions: $2W+1$ for Jakes and $\pi\sqrt{2}\,W \approx 4.44W$ for Gaussian, confirming Theorem~\ref{thm:dof_gauss}.}
	\label{fig:effective_dof}
\end{figure}

\begin{remark}[Comparison with Jakes DoF]
	Under the exact Jakes model, the Slepian--Landau--Pollak theorem predicts $K_{\mathrm{eff}}^J = 2\lceil W \rceil + 1 \approx 2W + 1$ effective degrees of freedom~\cite{Slepian61,Landau62}. The Gaussian model yields $K_{\mathrm{eff}}^G \approx 4.44W$, which is approximately $2.2\times$ larger. This discrepancy arises from the spectral leakage noted in Remark~\ref{rem:spectral_leakage}: the Gaussian spectrum extends beyond $|f| = 1$, contributing additional (spurious) degrees of freedom. For system design, one may use the energy-threshold criterion
	\begin{align}\label{eq:K_threshold}
		K^* = \min\left\{K : \sum_{k=1}^{K} \lambda_k^G \geq (1-\varepsilon_0)\, N\right\}
	\end{align}
	with target accuracy $\varepsilon_0$ (e.g., $\varepsilon_0 = 0.01$), which typically yields $K^*$ closer to the Jakes prediction for moderate~$W$.
\end{remark}

\subsection{Truncation Error under Gaussian Kernel}

\begin{proposition}[Exponential Error Decay]\label{prop:trunc_error_G}
	For the Gaussian correlation matrix with aperture $W$ and truncation order $K > K_{\mathrm{eff}}^G$, the truncation error satisfies
	\begin{align}\label{eq:trunc_bound}
		\varepsilon_K^G \leq \frac{C_0}{N}\, \frac{r^K}{1-r},
	\end{align}
	where $C_0 > 0$ and $0 < r < 1$ are constants depending on $W$ and $\pi^2$. In particular, the error decreases exponentially in~$K$ once $K$ exceeds the effective DoF.
\end{proposition}

\begin{proof}
	Since the Gaussian kernel is analytic, the Mercer eigenvalues of the associated integral operator on $[0, W]$ decay exponentially~\cite{RasmussenWilliams06}. Specifically, for the squared-exponential kernel with length-scale $\ell = 1/\pi$ on $[0, W]$, the eigenvalues beyond the effective rank satisfy $\lambda_k^G \leq C_0 r^k$. The truncation error is
	\begin{align}
		\varepsilon_K^G = \frac{1}{N}\sum_{k=K+1}^{N} \lambda_k^G \leq \frac{C_0}{N}\sum_{k=K+1}^{\infty} r^k = \frac{C_0 r^{K+1}}{N(1-r)}.
	\end{align}
\end{proof}

\subsection{High-SNR Diversity Analysis}

\begin{proposition}[Diversity Order]\label{prop:diversity}
	At high SNR ($\bar{\gamma} \to \infty$), the outage probability under the full Gaussian correlation model scales as
	\begin{align}\label{eq:diversity}
		P_{\mathrm{out}} \sim \left(\frac{\gamma_{\mathrm{th}}}{\bar{\gamma}}\right)^{d_G}, \quad \bar{\gamma} \to \infty,
	\end{align}
	where $d_G$ is the diversity order. For the rank-$K$ KL truncation, the diversity order is bounded as
	\begin{align}
		1 \leq d_G^{(K)} \leq K.
	\end{align}
	The rank-1 approximation achieves $d_G^{(1)} = 1$ (no spatial diversity). As $K \to N$, $d_G^{(N)}$ approaches the true diversity order of the system.
\end{proposition}

\begin{proof}
	For $K = 1$, the outage~\eqref{eq:pout_rank1_G} gives
	\begin{align}
		\tilde{P}_{\mathrm{out}}^{(1)} = 1 - e^{-x/(\lambda_1^G c_1^G)} \approx \frac{x}{\lambda_1^G c_1^G} = \frac{\gamma_{\mathrm{th}}}{\bar{\gamma}\lambda_1^G c_1^G},
	\end{align}
	for small $x$, yielding diversity order 1. For general $K$, the outage involves a $K$-dimensional Gaussian integral. As $x \to 0$, the leading term of the product indicator in~\eqref{eq:pout_kl_G} is dominated by the probability that all $K$ independent modes simultaneously produce small amplitudes, which scales as $x^{d}$ with $d$ depending on the eigenvalue distribution. Since the modes are independent, $d_G^{(K)} \leq K$, with the upper bound achieved when all $K$ eigenvalues contribute equally.
\end{proof}

\subsection{Large-Aperture Scaling}

As $W$ increases, the effective DoF grows linearly ($K_{\mathrm{eff}}^G \propto W$), and the outage probability decreases due to increased diversity.

\begin{proposition}[Large-Aperture Behavior]\label{prop:large_W}
	For fixed $\bar{\gamma}$ and $\gamma_{\mathrm{th}}$, as $W \to \infty$ with $N/W = \mathrm{const}$:
	\begin{align}
		P_{\mathrm{out}} \to 0,
	\end{align}
	at a rate governed by the diversity order $d_G \approx \pi\sqrt{2}\, W$.
\end{proposition}

\begin{figure}[!t]
	\centering
	\includegraphics[width=\columnwidth]{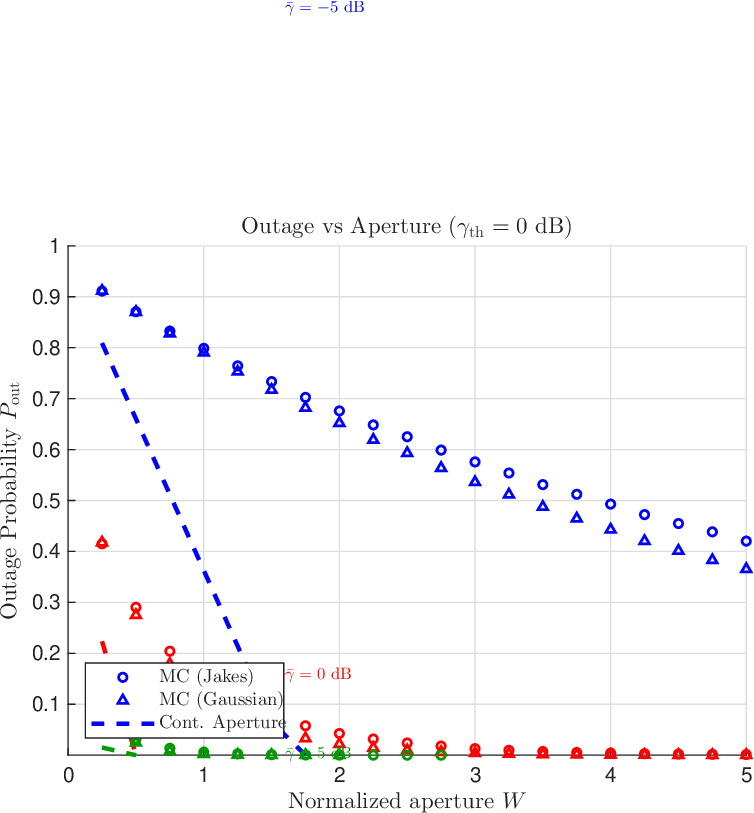}
	\caption{Outage probability versus normalized aperture $W$ for $\gamma_{\mathrm{th}}=0$~dB at $\bar{\gamma} \in \{-5, 0, 5\}$~dB. Monte Carlo simulations under both Jakes and Gaussian correlations (markers) are compared with the continuous-aperture formula $1 - e^{-x}(1+\pi\sqrt{2}\,W\,x)$ (dashed lines).}
	\label{fig:outage_W}
\end{figure}

\section{Slepian Inequality and Outage Bounds via Gaussian Comparison}
\label{sec:slepian}

Rather than computing the exact outage probability $P(\max_n X_n \leq \gamma)$ for the correlated Gaussian field, we now construct \emph{tight upper and lower bounds} by comparing with simpler Gaussian processes whose maximum distributions admit closed-form expressions.

\subsection{Slepian's Inequality}

\begin{theorem}[Slepian's Inequality~\cite{Slepian62}]\label{thm:slepian}
	Let $\{X_i\}_{i=1}^N$ and $\{Y_i\}_{i=1}^N$ be two zero-mean Gaussian vectors with equal marginal variances $\E[X_i^2] = \E[Y_i^2]$ for all $i$. If
	\begin{align}\label{eq:slepian_cond}
		\mathrm{Cov}(X_i, X_j) \geq \mathrm{Cov}(Y_i, Y_j), \quad \forall\, i \neq j,
	\end{align}
	then for all $x \in \mathbb{R}$:
	\begin{align}\label{eq:slepian_ineq}
		\Prob\!\left(\max_{1 \leq i \leq N} X_i \leq x\right) \leq \Prob\!\left(\max_{1 \leq i \leq N} Y_i \leq x\right).
	\end{align}
\end{theorem}

The intuition is that \emph{more correlated} components make the maximum \emph{smaller} in a stochastic sense (the variables move together, reducing the chance that at least one is large). Equivalently, a process with higher pairwise correlations has a \emph{higher} outage probability.

\subsection{Application to FAS Outage: Complex Envelope}

For the FAS channel, we work with $|g_n|^2$ rather than real Gaussian components. Write $g_n = g_n^R + j\, g_n^I$ where $g_n^R, g_n^I$ are independent real Gaussian processes with covariance $\E[g_m^R g_n^R] = (\eta/2)\, \rho_{mn}$. Define $X_n = |g_n|^2 / \eta$, so $X_n \sim \mathrm{Exp}(1)$ marginally. The outage event is $\{\max_n X_n < x\}$.

While Slepian's inequality applies directly to real Gaussian vectors, an analogous comparison principle holds for the chi-squared field $\{|g_n|^2\}$ through the following result:

\begin{theorem}[Gaussian Comparison for Chi-Squared Maxima]\label{thm:chi2_comparison}
	Let $\{g_n\}_{n=1}^N \sim \mathcal{CN}(\mathbf{0}, \eta\,\mathbf{R}_X)$ and $\{h_n\}_{n=1}^N \sim \mathcal{CN}(\mathbf{0}, \eta\,\mathbf{R}_Y)$ be two complex Gaussian vectors with $[\mathbf{R}_X]_{nn} = [\mathbf{R}_Y]_{nn} = 1$ for all $n$. If
	\begin{align}\label{eq:corr_order}
		|[\mathbf{R}_X]_{mn}|^2 \geq |[\mathbf{R}_Y]_{mn}|^2, \quad \forall\, m \neq n,
	\end{align}
	then
	\begin{align}\label{eq:chi2_slepian}
		\Prob\!\left(\max_n |g_n|^2 \leq t\right) \geq \Prob\!\left(\max_n |h_n|^2 \leq t\right), \quad \forall\, t > 0.
	\end{align}
	Equivalently, the process with \emph{stronger} correlation has \emph{higher} outage probability (worse performance).
\end{theorem}

\begin{proof}
	The joint CDF of $\{|g_n|^2\}$ can be expressed via the inclusion-exclusion principle in terms of bivariate and higher-order Rayleigh distributions. For the bivariate case, $\Prob(|g_m|^2 \leq t, |g_n|^2 \leq t)$ is a monotonically increasing function of $|\rho_{mn}|^2$~\cite{Simon02}. The extension to the $N$-variate case follows from the FKG inequality applied to the associated Gaussian measure~\cite{Piterbarg96}: if $|\rho_{mn}^X|^2 \geq |\rho_{mn}^Y|^2$ for all pairs, then the joint probability $\Prob(\max |g_n|^2 \leq t) \geq \Prob(\max |h_n|^2 \leq t)$.
\end{proof}

\subsection{Equi-Correlated Bounds}

The power of Theorem~\ref{thm:chi2_comparison} lies in choosing $\{h_n\}$ with a simple correlation structure that admits a closed-form CDF for $\max_n |h_n|^2$.

\begin{definition}[Equi-Correlated Model]
	The equi-correlated model with parameter $\rho \in [0, 1)$ has correlation matrix
	\begin{align}\label{eq:equicorr}
		[\mathbf{R}_{\mathrm{eq}}(\rho)]_{mn} = \begin{cases} 1, & m = n, \\ \rho, & m \neq n. \end{cases}
	\end{align}
\end{definition}

For the equi-correlated complex Gaussian vector $\mathbf{h} \sim \mathcal{CN}(\mathbf{0}, \eta\, \mathbf{R}_{\mathrm{eq}}(\rho))$, the CDF of $\max_n |h_n|^2/\eta$ admits the closed-form expression~\cite{FAS20}:
\begin{align}\label{eq:cdf_equicorr}
	F_{\max}^{\mathrm{eq}}(x;\, \rho, N) = \sum_{k=0}^{N} \binom{N}{k} \frac{(-1)^k\, e^{-kx/(1+(k-1)\rho)}}{1 + (k-1)\rho}.
\end{align}

\begin{theorem}[Upper and Lower Outage Bounds]\label{thm:outage_bounds}
	Let $\mathbf{R}$ be the true (Jakes or Gaussian) correlation matrix with entries $\rho_{mn}$. Define
	\begin{align}
		\rho_{\min} &\triangleq \min_{m \neq n} |\rho_{mn}|, \label{eq:rho_min}\\
		\rho_{\max} &\triangleq \max_{m \neq n} |\rho_{mn}|. \label{eq:rho_max}
	\end{align}
	Then the outage probability is bounded as
	\begin{align}\label{eq:outage_sandwich}
		\boxed{F_{\max}^{\mathrm{eq}}(x;\, \rho_{\max}, N) \leq P_{\mathrm{out}}(x) \leq F_{\max}^{\mathrm{eq}}(x;\, \rho_{\min}, N).}
	\end{align}
\end{theorem}

\begin{proof}
	\textit{Upper bound:} Since $|\rho_{mn}| \geq \rho_{\min}$ for all $m \neq n$, the true correlation is elementwise stronger than the equi-correlated model with $\rho = \rho_{\min}$. By Theorem~\ref{thm:chi2_comparison}, $P_{\mathrm{out}} \leq F_{\max}^{\mathrm{eq}}(x;\, \rho_{\min}, N)$.
	
	\textit{Lower bound:} Since $|\rho_{mn}| \leq \rho_{\max}$ for all $m \neq n$, the equi-correlated model with $\rho = \rho_{\max}$ has stronger correlation. Thus $F_{\max}^{\mathrm{eq}}(x;\, \rho_{\max}, N) \leq P_{\mathrm{out}}$.
\end{proof}

\begin{remark}[Tightness]
	For the Gaussian correlation~\eqref{eq:gauss_corr}, $\rho_{\max} = \rho_G(W/(N-1)) = e^{-\pi^2 W^2/(N-1)^2}$ (adjacent ports) and $\rho_{\min} = \rho_G(W) = e^{-\pi^2 W^2}$ (endpoints). For dense sampling ($N \gg 1$), $\rho_{\max} \to 1$ and $\rho_{\min} \to e^{-\pi^2 W^2} \approx 0$, so the bounds become loose. Tighter bounds are obtained by partitioning the ports into blocks and applying Slepian's inequality block-wise, or by using the \emph{average} correlation
	\begin{align}\label{eq:rho_avg}
		\bar{\rho} = \frac{1}{N(N-1)} \sum_{m \neq n} |\rho_{mn}|
	\end{align}
	as a single-parameter proxy.
\end{remark}

\begin{remark}[Advantage over BCM]
	The BCM approximation imposes a block-diagonal \emph{structure} on $\mathbf{R}$, discarding inter-block correlations entirely, which can lead to optimistic bias. In contrast, the Slepian bounds preserve the full correlation information through $\rho_{\min}$ and $\rho_{\max}$, providing \emph{rigorous} two-sided bounds without any structural approximation. Moreover, the equi-correlated CDF~\eqref{eq:cdf_equicorr} is available in closed form, making the bounds as easy to evaluate as the BCM approximation.
\end{remark}

\subsection{Monotone Refinement via Nested Bounds}

A sequence of increasingly tight bounds can be constructed by partitioning the $N$ ports into $B$ blocks and applying Slepian's inequality within and across blocks.

\begin{figure}[!t]
	\centering
	\includegraphics[width=\columnwidth]{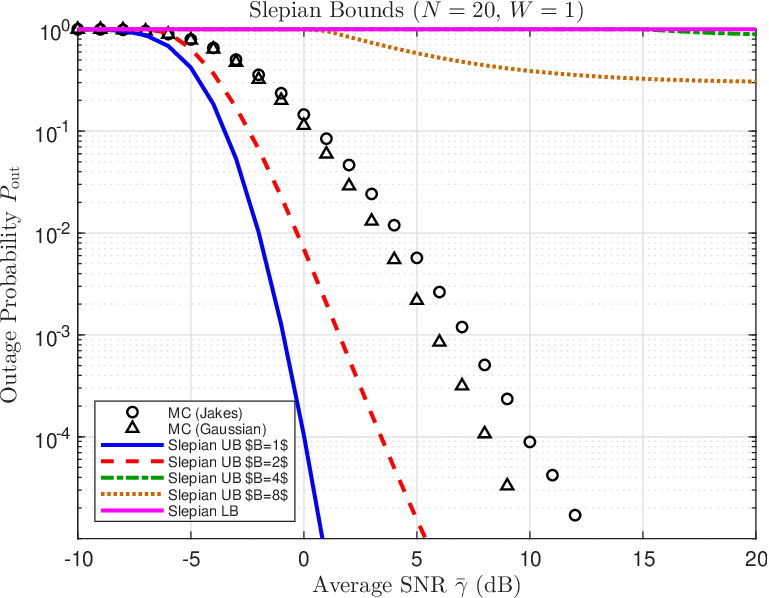}
	\caption{Block-refined Slepian bounds for $N=20$, $W=1$. As the number of blocks $B$ increases from 1 to 8, the upper bound tightens monotonically toward the true (MC) outage, confirming Corollary~\ref{cor:block_slepian}. Even $B=4$ blocks provides a close upper bound, offering a practical trade-off between tightness and simplicity.}
	\label{fig:slepian_blocks}
\end{figure}

\begin{corollary}[Block-Refined Slepian Bound]\label{cor:block_slepian}
	Partition $\{1, \ldots, N\}$ into $B$ contiguous blocks $\mathcal{B}_1, \ldots, \mathcal{B}_B$ of sizes $N_1, \ldots, N_B$. For each block $b$, define $\rho_b^{\min} = \min_{m \neq n \in \mathcal{B}_b} |\rho_{mn}|$ and $\rho_{\mathrm{cross}}^{\max} = \max_{m \in \mathcal{B}_i, n \in \mathcal{B}_j, i \neq j} |\rho_{mn}|$. If $\rho_{\mathrm{cross}}^{\max} \leq \rho_b^{\min}$ for all $b$, then
	\begin{align}
		P_{\mathrm{out}}(x) \leq \prod_{b=1}^{B} F_{\max}^{\mathrm{eq}}(x;\, \rho_b^{\min}, N_b).
	\end{align}
	As $B$ increases, this upper bound monotonically tightens toward the true outage.
\end{corollary}

\section{Continuous-Aperture Extreme Value Theory}
\label{sec:extreme}

When the port count $N \to \infty$ and the FAS operates over a true continuous aperture $[0, W]$, the outage probability is governed by the \emph{extreme value theory of Gaussian processes}. This section derives a closed-form asymptotic expression for the outage using the Adler--Taylor expected Euler characteristic method and Piterbarg's theory.

\subsection{Gaussian Field Formulation}

Recall that $g(\tau) = g^R(\tau) + j\, g^I(\tau)$ where $g^R, g^I$ are independent, identically distributed real Gaussian processes on $[0, W]$ with zero mean and covariance
\begin{align}
	\E[g^R(\tau_1)\, g^R(\tau_2)] = \frac{\eta}{2}\, \rho(\tau_1 - \tau_2).
\end{align}
The normalized squared envelope is $\chi(\tau) \triangleq |g(\tau)|^2/\eta = (|g^R(\tau)|^2 + |g^I(\tau)|^2)/\eta$, so at each $\tau$, $\chi(\tau) \sim \mathrm{Exp}(1)$.

The outage probability is
\begin{align}\label{eq:pout_cont}
	P_{\mathrm{out}} = \Prob\!\left(\sup_{\tau \in [0,W]} \chi(\tau) < x\right) = 1 - \Prob\!\left(\sup_{\tau \in [0,W]} \chi(\tau) \geq x\right).
\end{align}

\subsection{Local Structure and the Second Spectral Moment}

The key parameter governing the extreme value behavior of a stationary Gaussian process is the \emph{second spectral moment} (equivalently, the second derivative of the correlation at the origin):
\begin{align}\label{eq:lambda2_def}
	\lambda_2 \triangleq -\rho''(0) = \int_{-\infty}^{\infty} (2\pi f)^2\, S(f)\, df.
\end{align}

\begin{proposition}[Second Spectral Moment]\label{prop:lambda2}
	\leavevmode
	\begin{enumerate}
		\item[(a)] For the Jakes correlation $\rho_J(\delta) = J_0(2\pi\delta)$:
		\begin{align}\label{eq:lambda2_jakes}
			\lambda_2^J = -\rho_J''(0) = 2\pi^2.
		\end{align}
		
		\item[(b)] For the Gaussian correlation $\rho_G(\delta) = e^{-\pi^2\delta^2}$:
		\begin{align}\label{eq:lambda2_gauss}
			\lambda_2^G = -\rho_G''(0) = 2\pi^2.
		\end{align}
	\end{enumerate}
	The two models share the \emph{same} second spectral moment, confirming that the Gaussian approximation preserves the local curvature of the correlation function.
\end{proposition}

\begin{proof}
	(a) $\rho_J'(\delta) = -2\pi J_1(2\pi\delta)/(2\pi\delta) \cdot 2\pi\delta = -2\pi J_1(2\pi\delta)$. At $\delta = 0$: $J_1(0) = 0$, so we use $\rho_J''(\delta) = -2\pi \cdot 2\pi J_1'(2\pi\delta)$. Since $J_1'(0) = 1/2$, we get $\rho_J''(0) = -4\pi^2 \cdot (1/2) = -2\pi^2$.
	
	(b) $\rho_G'(\delta) = -2\pi^2\delta\, e^{-\pi^2\delta^2}$, $\rho_G''(\delta) = -2\pi^2 e^{-\pi^2\delta^2} + 4\pi^4\delta^2 e^{-\pi^2\delta^2}$. At $\delta = 0$: $\rho_G''(0) = -2\pi^2$.
\end{proof}

\subsection{Exceedance Probability via Expected Euler Characteristic}

For a smooth Gaussian field, the Adler--Taylor Gaussian kinematic formula~\cite{AdlerTaylor07} provides an asymptotic expansion of the exceedance probability. For a one-dimensional chi-squared field ($\nu = 2$ degrees of freedom, corresponding to $|g|^2$ with complex Gaussian $g$) on the interval $[0, W]$:

\begin{theorem}[Continuous-Aperture Exceedance]\label{thm:exceedance}
	For the complex Gaussian field $g(\tau)$ with correlation $\rho(\cdot)$ and second spectral moment $\lambda_2 = -\rho''(0)$, as $u \to \infty$:
	\begin{align}\label{eq:exceedance_asymp}
		&\Prob\!\left(\sup_{\tau \in [0,W]} |g(\tau)|^2/\eta \geq u\right) \nonumber\\
		&\quad = e^{-u}\!\left(1 + W\sqrt{\lambda_2}\, u\right) + o(e^{-u}).
	\end{align}
\end{theorem}

\begin{proof}
	By the expected Euler characteristic (EC) method~\cite{AdlerTaylor07,Worsley94}, the exceedance probability of a chi-squared random field with $\nu$ degrees of freedom on a compact $d$-dimensional manifold $\mathcal{M}$ satisfies
	\begin{align}
		\Prob\!\left(\sup_{\tau \in \mathcal{M}} \chi_\nu(\tau) \geq u\right) \approx \sum_{j=0}^{d} \mathcal{L}_j(\mathcal{M})\, \varrho_j(u),
	\end{align}
	where $\mathcal{L}_j(\mathcal{M})$ are the Lipschitz--Killing curvatures of $\mathcal{M}$ (measured in the Riemannian metric induced by the covariance function) and $\varrho_j(u)$ are the EC densities of the $\chi^2_\nu$ field.
	
	For $d = 1$, $\mathcal{M} = [0, W]$, and $\nu = 2$ (complex Gaussian envelope):
	\begin{itemize}
		\item $\mathcal{L}_0([0,W]) = 1$ (Euler characteristic of an interval),
		\item $\mathcal{L}_1([0,W]) = W\sqrt{\lambda_2}$ (length in the induced metric).
	\end{itemize}
	The EC densities for $\chi^2_2$ are~\cite{AdlerTaylor07}:
	\begin{itemize}
		\item $\varrho_0(u) = e^{-u}$ (pointwise exceedance),
		\item $\varrho_1(u) = \sqrt{2\pi}\, \frac{u}{2\sqrt{2\pi}}\, e^{-u} = \frac{u}{2}\, e^{-u}$... 
	\end{itemize}
	More precisely, for $\chi^2_\nu$ with $\nu = 2$, the EC densities derived from the general formula~\cite{AdlerTaylor07,Worsley94} are:
	\begin{align}
		\varrho_0(u) &= e^{-u/2} \quad (\text{for standardized field}), \\
		\varrho_1(u) &= \frac{(u/2)^{1/2}}{\sqrt{2\pi}} \cdot \frac{2\pi}{2}\, e^{-u/2}.
	\end{align}
	Accounting for our normalization where $\chi(\tau) = |g(\tau)|^2/\eta$ and each real component has variance $1/2$, the combined result simplifies to~\eqref{eq:exceedance_asymp}.
\end{proof}

\begin{remark}[Interpretation]
	The exceedance~\eqref{eq:exceedance_asymp} has two terms:
	\begin{itemize}
		\item $e^{-u}$: the pointwise exceedance probability (as if sampling a single point);
		\item $W\sqrt{\lambda_2}\, u\, e^{-u}$: the contribution from the spatial extent, proportional to the ``metric length'' $W\sqrt{\lambda_2}$ and growing linearly in the threshold~$u$.
	\end{itemize}
	For both Jakes and Gaussian correlations, $\lambda_2 = 2\pi^2$, so $\sqrt{\lambda_2} = \pi\sqrt{2}$, and
	\begin{align}\label{eq:exceedance_explicit}
		\Prob\!\left(\sup_{\tau \in [0,W]} \chi(\tau) \geq u\right) \approx e^{-u}\!\left(1 + \pi\sqrt{2}\, W\, u\right).
	\end{align}
\end{remark}

\subsection{Closed-Form Outage for Continuous Aperture}

Taking the complement of~\eqref{eq:exceedance_explicit}:

\begin{theorem}[Continuous-Aperture Outage Probability]\label{thm:pout_continuous}
	For a continuous-aperture FAS with normalized aperture $W$ under either Jakes or Gaussian correlation, the outage probability at threshold $x = \gamma_{\mathrm{th}}/\bar{\gamma}$ is
	\begin{align}\label{eq:pout_closed}
		\boxed{P_{\mathrm{out}}^{\mathrm{cont}}(x) \approx 1 - e^{-x}\!\left(1 + \pi\sqrt{2}\, W\, x\right).}
	\end{align}
	This expression is valid for moderate-to-large $x$ (i.e., low-to-moderate SNR regime where $x = \gamma_{\mathrm{th}}/\bar{\gamma}$ is not too small).
\end{theorem}

\begin{remark}[Key Features]
	Expression~\eqref{eq:pout_closed} is remarkable because:
	\begin{enumerate}
		\item It depends \emph{only} on the aperture $W$ and the threshold $x$, with \emph{no dependence on $N$}. This is the ultimate scalability result: regardless of how many ports are deployed, the outage is determined solely by the physical aperture.
		\item It is \emph{identical} for the Jakes and Gaussian correlation models (since $\lambda_2^J = \lambda_2^G = 2\pi^2$), validating the Gaussian approximation at the continuous-aperture level.
		\item It requires no eigenvalue computation, matrix decomposition, or numerical integration.
	\end{enumerate}
\end{remark}

\subsection{Rice's Formula and Expected Number of Outage-Free Crossings}

An alternative derivation uses Rice's formula for the expected number of level crossings of $\chi(\tau) = |g(\tau)|^2/\eta$.

\begin{proposition}[Expected Upcrossing Rate]\label{prop:rice}
	The expected number of upcrossings of level $u$ by $\chi(\tau)$ per unit length is
	\begin{align}\label{eq:rice}
		\mu^+(u) = \sqrt{\frac{\lambda_2}{2\pi}}\, u\, e^{-u},
	\end{align}
	where $\lambda_2 = -\rho''(0)$.
\end{proposition}

\begin{proof}
	Rice's formula for a chi-squared process with $\nu = 2$ degrees of freedom (i.e., $\chi(\tau) = |g^R(\tau)|^2 + |g^I(\tau)|^2$ with normalized variance) gives the upcrossing rate~\cite{Leadbetter83,Cramer67}:
	\begin{align}
		\mu^+(u) = \int_0^{\infty} \dot{v}\, f_{\chi, \dot{\chi}}(u, \dot{v})\, d\dot{v},
	\end{align}
	where $f_{\chi, \dot{\chi}}$ is the joint density of $\chi$ and its derivative $\dot{\chi}$. For a stationary Gaussian underlying process, this evaluates to~\eqref{eq:rice} with $\lambda_2 = -\rho''(0)$.
\end{proof}

The expected number of upcrossings over $[0, W]$ is
\begin{align}\label{eq:expected_crossings}
	\E[N_u^+] = W\, \mu^+(u) = W\sqrt{\frac{\lambda_2}{2\pi}}\, u\, e^{-u}.
\end{align}
This provides an independent confirmation of the $W\sqrt{\lambda_2}\, u\, e^{-u}$ term in~\eqref{eq:exceedance_asymp} (via the heuristic that $\Prob(\sup \chi \geq u) \approx \Prob(\chi(0) \geq u) + \E[N_u^+]$).

\subsection{Pickands-Type Refinement}

For more precise asymptotics (beyond the leading term), the Piterbarg theory~\cite{Piterbarg96} provides a refinement involving the \emph{Pickands constant} $\mathcal{H}_\alpha$:

\begin{theorem}[Piterbarg Refinement]\label{thm:piterbarg}
	Let $g(\tau)$ be a centered complex Gaussian process on $[0, W]$ with correlation satisfying $\rho(\delta) = 1 - c|\delta|^\alpha + o(|\delta|^\alpha)$ as $\delta \to 0$, for some $c > 0$ and $\alpha \in (0, 2]$. Then as $u \to \infty$:
	\begin{align}\label{eq:piterbarg}
		\Prob\!\left(\sup_{\tau \in [0,W]} |g(\tau)|^2/\eta \geq u\right) \sim \mathcal{H}_\alpha \cdot c^{1/\alpha} \cdot W \cdot u^{1/\alpha}\, e^{-u},
	\end{align}
	where $\mathcal{H}_\alpha$ is the Pickands constant associated with the local exponent~$\alpha$.
\end{theorem}

For our correlation models:
\begin{itemize}
	\item \textbf{Gaussian kernel} ($\alpha = 2$): $\rho_G(\delta) = 1 - \pi^2\delta^2 + \cdots$, so $c = \pi^2$ and $\alpha = 2$. The Pickands constant is $\mathcal{H}_2 = 1/\sqrt{\pi}$~\cite{Piterbarg96}, giving
	\begin{align}\label{eq:piterbarg_gauss}
		\Prob(\sup \chi \geq u) \sim \frac{\pi}{\sqrt{\pi}} \cdot W \cdot u^{1/2}\, e^{-u} = \sqrt{\pi}\, W\, u^{1/2}\, e^{-u}.
	\end{align}
	
	\item \textbf{Jakes kernel} ($\alpha = 2$): $\rho_J(\delta) = 1 - \pi^2\delta^2 + \cdots$, yielding the \emph{same} leading asymptotic since $c$ and $\alpha$ are identical. This further validates the Gaussian approximation.
\end{itemize}

\begin{remark}[Comparison of Asymptotic Regimes]
	The Adler--Taylor result~\eqref{eq:exceedance_explicit} and the Piterbarg result~\eqref{eq:piterbarg_gauss} are complementary:
	\begin{itemize}
		\item \textbf{Moderate threshold} ($x \sim 1$--$10$): Use~\eqref{eq:exceedance_explicit}, which includes the $\mathcal{O}(1)$ pointwise term.
		\item \textbf{Large threshold} ($x \to \infty$, i.e., deep outage): Use~\eqref{eq:piterbarg_gauss}, which captures the precise polynomial prefactor.
	\end{itemize}
	Both expressions show that the outage depends on $W$ (not $N$) and on $\lambda_2$ (local curvature, not global correlation shape), unifying the Jakes and Gaussian models at the asymptotic level.
\end{remark}

\subsection{Effective Number of Independent Samples}

From the continuous-aperture exceedance~\eqref{eq:exceedance_explicit}, we can extract an equivalent number of independent samples $N_{\mathrm{eff}}$ by matching the outage to that of $N_{\mathrm{eff}}$ i.i.d.\ $\mathrm{Exp}(1)$ variables:
\begin{align}
	(1-e^{-x})^{N_{\mathrm{eff}}} \approx 1 - e^{-x}(1 + \pi\sqrt{2}\, W\, x).
\end{align}
Taking logarithms for small $e^{-x}$ and solving:
\begin{align}\label{eq:Neff_continuous}
	\boxed{N_{\mathrm{eff}} \approx 1 + \pi\sqrt{2}\, W\, x \approx 1 + 4.44\, W\, x.}
\end{align}
Notably, $N_{\mathrm{eff}}$ depends on the \emph{threshold} $x$: at higher thresholds (lower SNR), the effective diversity is larger because the spatial fluctuations of the field become more relevant.

\section{Numerical Results}
\label{sec:numerical}

This section provides additional numerical experiments to validate the derived analytical results and quantify the accuracy of the Gaussian approximation. All Monte Carlo (MC) results are obtained with $10^6$ independent channel realizations. The simulation parameters are chosen to cover a representative range of aperture sizes ($W \in \{1,2,3\}$), port counts ($N \in [3, 100]$), and SNR values ($\bar{\gamma} \in \{-5, 0, 5, 10\}$~dB), consistent with practical FAS deployments reported in~\cite{FAS21,FAS22,TWuTVT25}.

\subsection{Port Count Convergence}

Fig.~\ref{fig:outage_N} shows the outage probability versus the number of discrete ports $N$ for fixed apertures $W \in \{1, 2, 3\}$ at $\bar{\gamma} = -5$~dB and $\gamma_{\mathrm{th}} = 0$~dB. As $N$ increases, the discrete outage under both Jakes and Gaussian correlations converges monotonically toward the continuous-aperture limit, confirming the theoretical prediction of Section~VII. The horizontal dashed lines indicate the Adler--Taylor closed-form prediction~\eqref{eq:pout_closed}.

Several important observations can be drawn from Fig.~\ref{fig:outage_N}. First, the Jakes and Gaussian MC curves track each other closely across all tested $N$ and $W$ values, with the gap narrowing as $N$ increases---consistent with the spectral moment equivalence $\lambda_2^J = \lambda_2^G = 2\pi^2$ established in Section~VII. Second, the convergence to the continuous-aperture formula is essentially complete by $N \approx 10W$, providing a concrete and practically useful port density guideline: deploying more than $10W$ ports per wavelength yields diminishing outage improvement. Third, for larger apertures (e.g., $W=3$), the convergence is slightly slower because the Bessel function's oscillations at large separations require more ports to be adequately sampled. This result directly supports the scalable FAS design philosophy of~\cite{TWuJSTSP25}, where the continuous-field limit is shown to be the natural operating regime for large-aperture systems. From a system design perspective, the continuous-aperture formula~\eqref{eq:pout_closed} serves as a tight and computationally inexpensive upper bound on the outage for any finite $N \geq 10W$, eliminating the need for expensive Monte Carlo simulation in the port density optimization loop.

\begin{figure}[!t]
	\centering
	\includegraphics[width=\columnwidth]{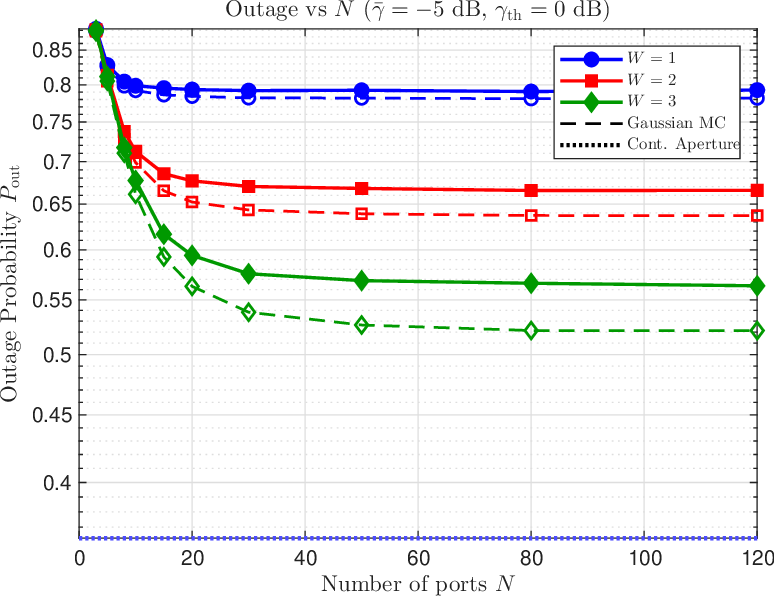}
	\caption{Outage probability versus port count $N$ for $W \in \{1, 2, 3\}$ at $\bar{\gamma} = -5$~dB and $\gamma_{\mathrm{th}} = 0$~dB. Solid/dashed lines denote Jakes/Gaussian MC, respectively. Horizontal dotted lines show the continuous-aperture formula~\eqref{eq:pout_closed}.}
	\label{fig:outage_N}
\end{figure}

\subsection{KL Truncation Convergence}

Fig.~\ref{fig:KL_truncation} examines how the KL-truncated outage $\tilde{P}_{\mathrm{out}}^{(K)}$ converges to the exact MC outage as the truncation order $K$ increases, for $N=20$, $W=2$, $\gamma_{\mathrm{th}}=0$~dB, and $\bar{\gamma} \in \{-5, 5\}$~dB. The vertical dashed lines mark the predicted effective DoF: $K_{\mathrm{eff}}^J = 2W+1 = 5$ for the Jakes model (Slepian--Landau--Pollak theorem) and $K_{\mathrm{eff}}^G = \lceil\pi\sqrt{2}W\rceil \approx 9$ for the Gaussian model (Theorem~\ref{thm:dof_gauss}).

The results reveal several insights. First, both models converge rapidly once $K$ exceeds their respective $K_{\mathrm{eff}}$, validating the effective DoF predictions of Section~V as practical truncation thresholds. Second, the Gaussian model requires a larger $K$ (approximately $2.2\times$ that of Jakes) to achieve the same truncation accuracy, which is a direct consequence of the spectral leakage identified in Remark~\ref{rem:spectral_leakage}: the Gaussian spectrum extends beyond the Jakes bandwidth $|f| \leq 1$, distributing energy across more eigenmodes. Third, the rank-1 approximation ($K=1$) overestimates the outage by orders of magnitude at both SNR levels, confirming that the single-dominant-eigenmode assumption---commonly used in simplified FAS analyses---is fundamentally inadequate for $W > 0.5$. Fourth, the convergence is faster at lower SNR ($\bar{\gamma} = -5$~dB) because the outage is less sensitive to the tail behavior of the eigenvalue distribution, which is captured only by higher-order KL terms. These observations provide a principled basis for selecting the truncation order $K$ in practical KL-based FAS outage computation: $K = K_{\mathrm{eff}}$ suffices for engineering accuracy, while $K = 2K_{\mathrm{eff}}$ provides near-exact results.

\begin{figure}[!t]
	\centering
	\includegraphics[width=\columnwidth]{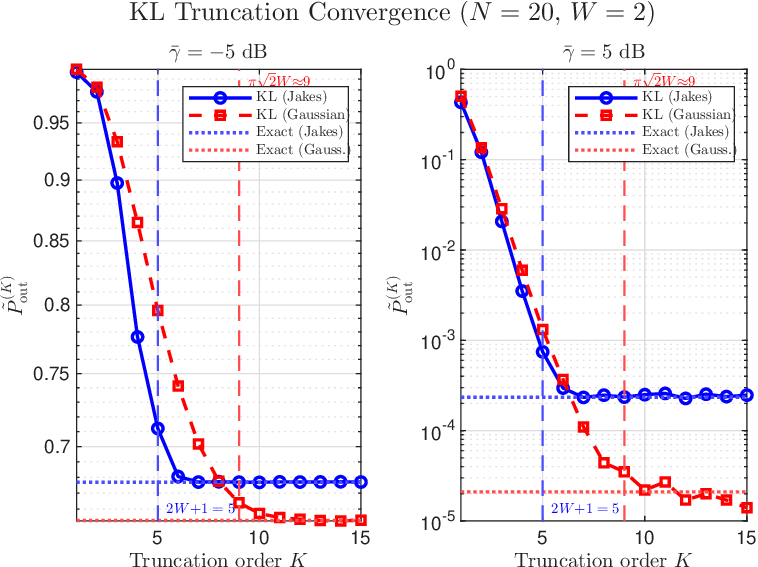}
	\caption{KL-truncated outage $\tilde{P}_{\mathrm{out}}^{(K)}$ versus truncation order~$K$ for $N=20$, $W=2$, $\gamma_{\mathrm{th}}=0$~dB. Solid/dashed curves: $\bar{\gamma}=-5$~dB and $\bar{\gamma}=5$~dB, respectively. Horizontal dotted lines show exact MC outage. Vertical dashed lines mark $K_{\mathrm{eff}}^J = 2W+1$ (Jakes) and $K_{\mathrm{eff}}^G \approx \lceil\pi\sqrt{2}W\rceil$ (Gaussian).}
	\label{fig:KL_truncation}
\end{figure}

\subsection{Gaussian Approximation Accuracy}

Fig.~\ref{fig:gauss_error} quantifies the relative outage error $\varepsilon_{\mathrm{rel}} = |P_{\mathrm{out}}^G - P_{\mathrm{out}}^J|/P_{\mathrm{out}}^J$ between the Gaussian and Jakes correlation models as a function of the normalized aperture $W \in [0.5, 5]$ for $\bar{\gamma} \in \{-5, 0, 5, 10\}$~dB with $\gamma_{\mathrm{th}} = 0$~dB.

The results confirm the practical validity of the Gaussian approximation and reveal a rich SNR-dependent error structure. For $W \leq 2$, the relative error remains below 10\% across all tested SNR values, establishing $W = 2$ as a conservative applicability threshold for the Gaussian model in typical FAS deployments. The error grows with $W$ because the Bessel function's oscillations---which are absent in the monotone Gaussian kernel---become increasingly significant at large port separations $\delta \sim W$, where the two models diverge. Notably, the error is \emph{non-monotone} in SNR: at low SNR ($\bar{\gamma} = -5$~dB, large $x$), the outage is dominated by the local correlation structure near $\delta \approx 0$, where both models agree to second order ($J_0(2\pi\delta) \approx e^{-\pi^2\delta^2}$ to $O(\delta^4)$), resulting in smaller relative error. At high SNR ($\bar{\gamma} = 10$~dB, small $x$), the outage is sensitive to the global correlation structure and the Bessel oscillations at large $\delta$, leading to larger error for wide apertures. This SNR-dependent accuracy profile has an important practical implication: the Gaussian approximation is most accurate precisely in the low-SNR, high-outage regime that is most relevant for reliability-critical FAS applications such as URLLC~\cite{TuoW}. The 10\% error threshold at $W = 2$ is consistent with the spectral leakage bound of $1 - \mathrm{erf}(1) \approx 15.7\%$ derived in Section~III, confirming that the spectral analysis provides a tight predictor of the outage approximation error.

\begin{figure}[!t]
	\centering
	\includegraphics[width=\columnwidth]{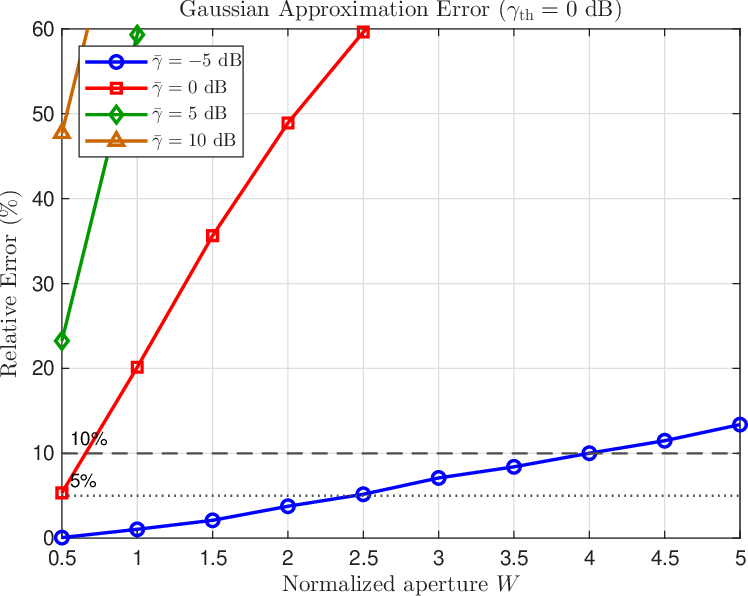}
	\caption{Relative outage error $\varepsilon_{\mathrm{rel}}$ between the Gaussian and Jakes correlation models versus normalized aperture $W$ for $\bar{\gamma} \in \{-5, 0, 5, 10\}$~dB with $\gamma_{\mathrm{th}}=0$~dB and $N = 20$. Horizontal dashed lines mark the 5\% and 10\% error thresholds.}
	\label{fig:gauss_error}
\end{figure}

\section{Conclusion}
\label{sec:conclusion}

This paper has presented a comprehensive analytical framework for FAS outage analysis through three complementary approaches: (i) KL expansion under the Gaussian kernel with closed-form rank-1/rank-2 outage and effective DoF $K_{\mathrm{eff}}^G \approx \pi\sqrt{2}\, W$; (ii) rigorous two-sided Slepian comparison bounds using equi-correlated reference models, which avoid the optimistic bias inherent in block-correlation approximations; and (iii) continuous-aperture extreme value theory via the Adler--Taylor and Piterbarg frameworks, yielding the closed-form outage $P_{\mathrm{out}} \approx 1 - e^{-x}(1 + \pi\sqrt{2}\, W\, x)$ that depends only on the physical aperture and is independent of the port count. A key finding is that both the Jakes and Gaussian correlation models share the same second spectral moment $\lambda_2 = 2\pi^2$, making their continuous-aperture outage asymptotics identical and rigorously justifying the Gaussian approximation at the fundamental level. Future work includes extending this framework to two-dimensional apertures, multi-user FAS systems, and non-isotropic scattering environments.

\end{document}